\documentclass[preprint]{elsarticle}

\usepackage{lineno,hyperref}
\usepackage{graphicx,amsmath,amssymb,amsthm,bm}
\modulolinenumbers[5]

\journal{}

\newtheorem{definition}{Definition}
\newtheorem{proposition}{Proposition}

\newtheorem{theorem}{Theorem}
\newtheorem{corollary}{Corollary}
\newtheorem{assumption}{Assumption}

\begin{document}

\begin{frontmatter}

\title{Memory-two strategies forming symmetric mutual reinforcement learning equilibrium in repeated prisoners' dilemma game}

\author{Masahiko Ueda\corref{mycorrespondingauthor}}
\address{Graduate School of Sciences and Technology for Innovation, Yamaguchi University, Yamaguchi 753-8511, Japan}
\ead{m.ueda@yamaguchi-u.ac.jp}

\begin{abstract}
We investigate symmetric equilibria of mutual reinforcement learning when both players alternately learn the optimal memory-two strategies against the opponent in the repeated prisoners' dilemma game.
We provide a necessary condition for memory-two deterministic strategies to form symmetric equilibria.
We then provide three examples of memory-two deterministic strategies which form symmetric mutual reinforcement learning equilibria.
We also prove that mutual reinforcement learning equilibria formed by memory-two strategies are also mutual reinforcement learning equilibria when both players use reinforcement learning of memory-$n$ strategies with $n>2$.
\end{abstract}


\end{frontmatter}


\section{Introduction}
\label{sec:introduction}
The prisoners' dilemma is one of the simplest situations in which rational actions of individuals do not maximize social welfare \cite{RCO1965}.
Although the best action of each agent is defection, mutual cooperation improves the utility of both agents.
On the other hand, if the prisoners' dilemma game is infinitely repeated, the situation changes.
In fact, mutual cooperation can be realized by rational behavior of each agent, and this result is known as folk theorem \cite{MaiSam2006}.
The folk theorem was also extended to a stronger version that any individually rational payoffs can be realized as subgame perfect equilibria \cite{FugMas1986}.

At the same time, it has been pointed out by experiments that the realistic agents like human beings are not necessarily rational, and theories of bounded rationality have been needed \cite{Aum1997}.
One of the mainstream is modeling agents by finite automata (agents with finite complexity) \cite{Ney1985,Rub1986,KalSta1988,AbrRub1988,BanSun1990,Ben1993,Ney1998}.
Especially, Abreu and Rubinstein found that the equilibrium payoffs realized by finite automaton selection games, where players choose finite automata as their strategies in repeated games so as to maximize their payoffs and to minimize the number of states of the finite automata lexicographically, are restricted to some small region in individually rational payoffs \cite{AbrRub1988}.
Kalai and Stanford proved that every subgame perfect equilibrium of repeated games can be approximated by a subgame perfect $\epsilon$-equilibrium of finite complexity \cite{KalSta1988}.
A slightly different approach from finite automata is modeling agents by ones with finite memory, which recall only a finite number past periods \cite{Leh1988}.
(Although there is distinction between memory and recall in computer science, we use these two words interchangeably.)
Deterministic finite-memory strategies are contained in a class of finite automata.
Sabourian and co-workers investigated how the folk theorem can be extended to finite-memory strategies \cite{Sab1998,BCS2009,BCS2016}.

Another trend of studies of bounded rationality is modeling agents as adaptive ones which gradually acquire favorable strategies.
One of the most successful approach is evolutionary game theory, where a population of individuals evolves by natural selection \cite{SmiPri1973}.
The concept of evolutionarily stable strategy, which is interpreted as stability against mutation, succeeded in strengthening the concept of Nash equilibrium.
However, it was also shown that any strategy in the infinitely repeated prisoners' dilemma game is not an evolutionarily stable strategy, and is not stable against neutral drift \cite{BoyLor1987}.
There are also studies of evolutionarily stable strategies with finite complexity \cite{FudMas1990,BinSam1992,NSE1995}.
Particularly, Binmore and Samuelson proposed a modified version of evolutionarily stable strategy and showed that such strategies must maximize the sum of payoffs of two players \cite{BinSam1992}.
Furthermore, many evolutionary simulations on finite-memory strategies have been done for various population sizes, mutation rates, and types of interaction \cite{NowSig1992,NowSig1993,BerLac2003,IFN2005,SPS2009,PGSFM2013}.
Stewart and Plotkin proposed the concept to evolutionary robust strategies, which is an extension of evolutionarily stable strategies to systems of finite population size and cannot be selectively replaced by any mutant strategies \cite{StePlo2013}.

Learning is another way of adaptation of human beings, and has also attracted much attention in theoretical economics \cite{KalLeh1993,FudLev1993,HarMas2000}, computer science \cite{Rou2016}, and complex systems theory \cite{KraKra1989,BisNai2000,SAF2002,MacFla2002,MasNak2011,GalFar2013}.
Many methods of learning have been proposed in game theory \cite{FugLev1998}, and compared with experimental results \cite{EreRot1998,Dal2005,DalFre2011}.
One of the most popular learning methods is reinforcement learning \cite{SutBar2018}.
In reinforcement learning, an agent gradually learns the optimal policy against a stationary environment.
Mutual reinforcement learning in game theory is a more difficult problem since the existence of multiple agents makes an environment nonstationary \cite{Rap1967,SanCri1996,HuWel2003,HKJKGC2017,BDK2019}.
Several methods have been proposed for reinforcement learning with multiple agents \cite{BBD2008}.

Recently, memory-$n$ strategies ($n$ periods memory strategies) with $n>1$ attract much attention in computational evolutionary game theory, because longer memory enables more complicated behavior \cite{LiKen2013,YBC2017,HMCN2017,MurBae2018,MurBae2020,Ued2021}.
Especially, longer memory enables us to design robust strategies against implementation errors.
Since agents in evolutionary biology are organisms, which are far from rational, it has been traditionally assumed that the length of memory of such agents is assumed to be short.
This is in contrast to chronology of game theory in economics, where behaviors of rational and forward-looking agents were first studied and then memory length becomes shorter in order to describe agents with bounded rationality.
Because rationality of realistic agents is bounded, shorter-memory strategies will be preferred if complexity is also considered.

Here, we investigate mutual reinforcement learning in the repeated prisoners' dilemma game \cite{RCO1965}.
More explicitly, we investigate properties of equilibria formed by learning agents when the two agents alternately learn their optimal strategies against the opponent.
In the previous study \cite{UsuUed2021}, it was found that, among all deterministic memory-one strategies, only the Grim trigger strategy, the Win-Stay Lose-Shift strategy, and the All-$D$ strategy can form symmetric equilibrium of mutual reinforcement learning.
A natural question is ``How does the set of such equilibria grow as the length of memory increases?''.
Such direction of research can be useful when we construct strong strategies based on memory-one strategies, as in computational evolutionary game theory.
Furthermore, we want to understand mutual reinforcement learning equilibria in terms of strategies, not equilibrium payoffs.
However, even whether the above equilibria formed by memory-one strategies are still equilibria in memory-$n$ settings or not has not been known.

In this paper, we extend the analysis of Ref. \cite{UsuUed2021} to memory-two strategies.
First, we provide a necessary condition for memory-two deterministic strategies to form symmetric equilibria.
Then we provide three non-trivial examples of memory-two deterministic strategies which form symmetric mutual reinforcement learning equilibria.
Furthermore, we also prove that mutual reinforcement learning equilibria formed by memory-$n^\prime$ strategies are also mutual reinforcement learning equilibria when both players use reinforcement learning of memory-$n$ strategies with $n>n^\prime$.

This paper is organized as follows.
In Section \ref{sec:model}, we introduce the repeated prisoners' dilemma game with memory-$n$ strategies, and players using reinforcement learning.
In Section \ref{sec:structure}, we show that the structure of the optimal strategies is constrained by the Bellman optimality equation.
In Section \ref{sec:symmetric}, we introduce the concepts of mutual reinforcement learning equilibrium and symmetric equilibrium.
We then provide a necessary condition for memory-two deterministic strategies to form symmetric equilibria.
In Section \ref{sec:example}, we provide three examples of memory-two deterministic strategies which form symmetric mutual reinforcement learning equilibria.
In Section \ref{sec:optimality}, we show that mutual reinforcement learning equilibria formed by memory-$n^\prime$ strategies are also mutual reinforcement learning equilibria when both players use reinforcement learning of memory-$n$ strategies with $n>n^\prime$.
Section \ref{sec:conclusion} is devoted to conclusion.

\section{Model}
\label{sec:model}
We introduce the repeated prisoners' dilemma game \cite{Rap1967}.
There are two players ($1$ and $2$) in the game.
Each player chooses cooperation ($C$) or defection ($D$) on every round.
The action of player $a$ is written as $\sigma_a \in \{ C, D \}$.
We collectively write $\bm{\sigma}:=\left( \sigma_1, \sigma_2 \right)$, and call $\bm{\sigma}$ an action profile.
We also write the space of all possible action profiles as $\Omega:=\{ C, D \}^2$.
The payoff of player $a\in \{ 1, 2 \}$ when the action profile is $\bm{\sigma}$ is described as $r_a\left( \bm{\sigma} \right)$.
The payoffs in the prisoners' dilemma game are given by
\begin{eqnarray}
 \left( r_1 \left( C, C \right), r_1 \left( C, D \right), r_1 \left( D, C \right), r_1 \left( D, D \right) \right) &=& (R, S, T, P) \\
 \left( r_2 \left( C, C \right), r_2 \left( C, D \right), r_2 \left( D, C \right), r_2 \left( D, D \right) \right) &=& (R, T, S, P)
\end{eqnarray}
with $T>R>P>S$ and $2R>T+S$.
The (time-independent) memory-$n$ strategy $(n\geq 1)$ of player $a$ is described as the conditional probability $T_a\left( \left. \sigma_a \right| \left[ \bm{\sigma}^{(-m)} \right]_{m=1}^{n} \right)$ of taking action $\sigma_a$ when the action profiles in the previous $n$ rounds are $\left[ \bm{\sigma}^{(-m)} \right]_{m=1}^{n}$, together with an initial condition, where we have introduced the notation $\left[ \bm{\sigma}^{(-m)} \right]_{m=1}^{n}:=\left( \bm{\sigma}^{(-1)}, \cdots, \bm{\sigma}^{(-n)} \right)$ from newest to oldest \cite{Ued2021}.
(As a strategy of bounded rational players, we use finite-memory strategies, not finite automata, because the former allows strategies to be stochastic.
Although stochastic strategies are allowed in our framework, we investigate only deterministic strategies in this paper.)
We write the length of memory of player $a$ as $n_a$ and define $n:= \max{\{n_1, n_2 \}}$.
In this paper, we assume that $n$ is finite.
\begin{assumption}
\label{ass:memory}
Both players use time-independent finite-memory strategies.
\end{assumption}
Below we introduce the notation $-a := \{ 1, 2 \} \backslash a$.

We consider the situation that both players learn their optimal strategies against the strategy of the opponent by reinforcement learning \cite{SutBar2018}.
In reinforcement learning, each player learns mapping (called policy) from the action profiles $\left[ \bm{\sigma}^{(-m)} \right]_{m=1}^{n}$ in the previous $n$ rounds to his/her action $\sigma$ so as to maximize his/her expected future reward.
We write the action of player $a$ at round $t$ as $\sigma_a(t)$.
In addition, we write $r_a(t):= r_a\left( \bm{\sigma}(t) \right)$.
We define the action-value function of player $a$ as
\begin{eqnarray}
 Q_a\left( \sigma_a, \left[ \bm{\sigma}^{(-m)} \right]_{m=1}^{n} \right) &:=& \mathbb{E} \left[ \left. \sum_{k=0}^\infty \gamma^k r_a(t+k) \right| \sigma_a(t)= \sigma_a, \left[ \bm{\sigma}(s) \right]_{s=t-1}^{t-n}=\left[ \bm{\sigma}^{(-m)} \right]_{m=1}^{n} \right], \nonumber \\
 &&
 \label{eq:def_Q}
\end{eqnarray}
where $\gamma$ is a discounting factor satisfying $0\leq \gamma < 1$.
The action-value function $Q_a\left( \sigma_a, \left[ \bm{\sigma}^{(-m)} \right]_{m=1}^{n} \right)$ represents the expected future payoffs $\sum_{k=0}^\infty \gamma^k r_a(t+k)$ of player $a$ after round $t$ by taking action $\sigma_a$ when action profiles in the previous $n$ rounds are $\left[ \bm{\sigma}^{(-m)} \right]_{m=1}^{n}$.
Therefore, the action-value function suggests the best action in each action profile.
It should be noted that the right-hand side does not depend on $t$.
Due to the property of memory-$n$ strategies, the action-value function $Q_a$ obeys the Bellman equation against a fixed strategy $T_{-a}$ of the opponent:
\begin{eqnarray}
 && Q_a\left( \sigma_a, \left[ \bm{\sigma}^{(-m)} \right]_{m=1}^{n} \right) \nonumber \\
 &=& \sum_{\sigma_{-a}} r_a \left( \bm{\sigma} \right) T_{-a} \left( \left. \sigma_{-a} \right| \left[ \bm{\sigma}^{(-m)} \right]_{m=1}^{n} \right) \nonumber \\
 && + \gamma \sum_{\sigma_{a}^{\prime}} \sum_{\sigma_{-a}} T_{a} \left( \left. \sigma_{a}^{\prime} \right| \bm{\sigma}, \left[ \bm{\sigma}^{(-m)} \right]_{m=1}^{n-1} \right) T_{-a} \left( \left. \sigma_{-a} \right| \left[ \bm{\sigma}^{(-m)} \right]_{m=1}^{n} \right) Q_a \left( \sigma_{a}^{\prime}, \bm{\sigma}, \left[ \bm{\sigma}^{(-m)} \right]_{m=1}^{n-1} \right). \nonumber \\
 &&
 \label{eq:BE}
\end{eqnarray}
See \ref{app:BE} for the derivation of Eq. (\ref{eq:BE}).
It has been known that the optimal policy $T^*_a$ and the optimal action-value function $Q^*_a$ obeys the following Bellman optimality equation:
\begin{eqnarray}
 && Q^*_a\left( \sigma_a, \left[ \bm{\sigma}^{(-m)} \right]_{m=1}^{n} \right) \nonumber \\
 &=& \sum_{\sigma_{-a}} r_a \left( \bm{\sigma} \right) T_{-a} \left( \left. \sigma_{-a} \right| \left[ \bm{\sigma}^{(-m)} \right]_{m=1}^{n} \right) \nonumber \\
 && + \gamma \sum_{\sigma_{-a}} T_{-a} \left( \left. \sigma_{-a} \right| \left[ \bm{\sigma}^{(-m)} \right]_{m=1}^{n} \right) \max_{\hat{\sigma}} Q^*_a \left( \hat{\sigma}, \bm{\sigma}, \left[ \bm{\sigma}^{(-m)} \right]_{m=1}^{n-1} \right),
 \label{eq:bellman}
\end{eqnarray}
with the support
\begin{eqnarray}
 \mathrm{supp}T^*_a\left( \left. \cdot \right| \left[ \bm{\sigma}^{(-m)} \right]_{m=1}^{n} \right) &=& \arg \max_{\sigma} Q^*_a \left( \sigma, \left[ \bm{\sigma}^{(-m)} \right]_{m=1}^{n} \right).
 \label{eq:BOE_support}
\end{eqnarray}
See \ref{app:BOE} for the derivation of Eqs. (\ref{eq:bellman}) and (\ref{eq:BOE_support}).
In other words, in the optimal policy against $T_{-a}$, player $a$ takes the action $\sigma_a$ which maximizes the value of $Q^*_a\left( \cdot, \left[ \bm{\sigma}^{(-m)} \right]_{m=1}^{n} \right)$ when the action profiles at the previous $n$ rounds are $\left[ \bm{\sigma}^{(-m)} \right]_{m=1}^{n}$.
In Q-learning, which is one of the simplest algorithms of reinforcement learning, it is known that values of action-value functions converge to the solutions of the Bellman optimality equation if all state-action pairs are visited an infinite number of times \cite{SutBar2018}.

We investigate the situation that players infinitely repeat the infinitely-repeated games and players alternately learn their optimal strategies in each game, as in Ref. \cite{UsuUed2021}.
We write the optimal strategy and the corresponding optimal action-value function of player $a$ at $d$-th game as $T^{*(d)}_a$ and $Q^{*(d)}_a$, respectively.
Given an initial strategy $T^{*(0)}_2$ of player $2$, in the $(2l-1)$-th game $(l\in \mathbb{N})$, player $1$ learns $T^{*(2l-1)}_1$ against $T^{*(2l-2)}_2$ by calculating $Q^{*(2l-1)}_1$.
In the $2l$-th game, player $2$ learns $T^{*(2l)}_2$ against $T^{*(2l-1)}_1$ by calculating $Q^{*(2l)}_2$.
We are interested in the fixed points of the dynamics, that is, $T^{*(\infty)}_a$ and $Q^{*(\infty)}_a$.

In this paper, we mainly investigate situations that the support (\ref{eq:BOE_support}) contains only one action, that is, strategies are deterministic.
The number of deterministic memory-$n$ strategies in the repeated prisoners' dilemma game is $2^{2^{2n}}$, which increases rapidly as $n$ increases.

\section{Structure of optimal strategies}
\label{sec:structure}
Below we consider only the case $n=2$.
The Bellman optimality equation (\ref{eq:bellman}) for $n=2$ is 
\begin{eqnarray}
 && Q^*_a\left( \sigma_a, \bm{\sigma}^{(-1)}, \bm{\sigma}^{(-2)} \right) \nonumber \\
 &=& \sum_{\sigma_{-a}} r_a \left( \bm{\sigma} \right) T_{-a} \left( \left. \sigma_{-a} \right| \bm{\sigma}^{(-1)}, \bm{\sigma}^{(-2)} \right) \nonumber \\
 && + \gamma \sum_{\sigma_{-a}} T_{-a} \left( \left. \sigma_{-a} \right| \bm{\sigma}^{(-1)}, \bm{\sigma}^{(-2)}\right) \max_{\hat{\sigma}} Q^*_a \left( \hat{\sigma}, \bm{\sigma}, \bm{\sigma}^{(-1)} \right)
 \label{eq:BOE_mt}
\end{eqnarray}
with
\begin{eqnarray}
 \mathrm{supp}T^*_a\left( \left. \cdot \right| \bm{\sigma}^{(-1)}, \bm{\sigma}^{(-2)} \right) &=& \arg \max_{\sigma} Q^*_a \left( \sigma, \bm{\sigma}^{(-1)}, \bm{\sigma}^{(-2)} \right).
 \label{eq:BOE_support_mt}
\end{eqnarray}
The number of memory-two deterministic strategies is $2^{16}$, which is quite large, and therefore we cannot investigate all memory-two deterministic strategies as in the case of memory-one deterministic strategies \cite{UsuUed2021}.
Instead, we first investigate general properties of optimal strategies.

We introduce the matrix representation of a strategy:
\begin{eqnarray}
 && T_a\left( \sigma \right) \nonumber \\
 &:=& \left(
\begin{array}{cccc}
T_a\left( \left. \sigma \right| (C,C), (C,C)\right) & T_a\left( \left. \sigma \right| (C,C), (C,D)\right) & T_a\left( \left. \sigma \right| (C,C), (D,C)\right) & T_a\left( \left. \sigma \right| (C,C), (D,D)\right) \\
T_a\left( \left. \sigma \right| (C,D), (C,C)\right) & T_a\left( \left. \sigma \right| (C,D), (C,D)\right) & T_a\left( \left. \sigma \right| (C,D), (D,C)\right) & T_a\left( \left. \sigma \right| (C,D), (D,D)\right) \\
T_a\left( \left. \sigma \right| (D,C), (C,C)\right) & T_a\left( \left. \sigma \right| (D,C), (C,D)\right) & T_a\left( \left. \sigma \right| (D,C), (D,C)\right) & T_a\left( \left. \sigma \right| (D,C), (D,D)\right) \\
T_a\left( \left. \sigma \right| (D,D), (C,C)\right) & T_a\left( \left. \sigma \right| (D,D), (C,D)\right) &T_a\left( \left. \sigma \right| (D,D), (D,C)\right) & T_a\left( \left. \sigma \right| (D,D), (D,D)\right)
\end{array}
\right). \nonumber \\
&&
\label{eq:stmatrix}
\end{eqnarray}
For deterministic strategies, each component in the matrix is $0$ or $1$.

We now prove the following proposition:
\begin{proposition}
\label{prop:structure}
For two different action profiles $\bm{\sigma}^{(-2)}$ and $\bm{\sigma}^{(-2)\prime}$, if
\begin{eqnarray}
 T_{-a} \left( \left. \sigma \right| \bm{\sigma}^{(-1)}, \bm{\sigma}^{(-2)} \right) &=& T_{-a} \left( \left. \sigma \right| \bm{\sigma}^{(-1)}, \bm{\sigma}^{(-2)\prime} \right) \quad \left( \forall \sigma \right)
\end{eqnarray}
holds for some $\bm{\sigma}^{(-1)}$, then 
\begin{eqnarray}
 T^*_a \left( \left. \sigma \right| \bm{\sigma}^{(-1)}, \bm{\sigma}^{(-2)} \right) &=& T^*_a \left( \left. \sigma \right| \bm{\sigma}^{(-1)}, \bm{\sigma}^{(-2)\prime} \right) \quad \left( \forall \sigma \right)
 \label{eq:structure}
\end{eqnarray}
also holds.
\end{proposition}

\begin{proof}
For such $\bm{\sigma}^{(-1)}$, because of Eq. (\ref{eq:BOE_mt}), we find
\begin{eqnarray}
 && Q^*_a\left( \sigma_a, \bm{\sigma}^{(-1)}, \bm{\sigma}^{(-2)} \right) \nonumber \\
 &=& \sum_{\sigma_{-a}} r_a \left( \bm{\sigma} \right) T_{-a} \left( \left. \sigma_{-a} \right| \bm{\sigma}^{(-1)}, \bm{\sigma}^{(-2)} \right) \nonumber \\
 && + \gamma \sum_{\sigma_{-a}} T_{-a} \left( \left. \sigma_{-a} \right| \bm{\sigma}^{(-1)}, \bm{\sigma}^{(-2)}\right) \max_{\hat{\sigma}} Q^*_a \left( \hat{\sigma}, \bm{\sigma}, \bm{\sigma}^{(-1)} \right) \nonumber \\
 &=& \sum_{\sigma_{-a}} r_a \left( \bm{\sigma} \right) T_{-a} \left( \left. \sigma_{-a} \right| \bm{\sigma}^{(-1)}, \bm{\sigma}^{(-2)\prime} \right) \nonumber \\
 && + \gamma \sum_{\sigma_{-a}} T_{-a} \left( \left. \sigma_{-a} \right| \bm{\sigma}^{(-1)}, \bm{\sigma}^{(-2)\prime}\right) \max_{\hat{\sigma}} Q^*_a \left( \hat{\sigma}, \bm{\sigma}, \bm{\sigma}^{(-1)} \right) \nonumber \\
 &=& Q^*_a\left( \sigma_a, \bm{\sigma}^{(-1)}, \bm{\sigma}^{(-2)\prime} \right)
\end{eqnarray}
for all $\sigma_a$.
Since $T^*_a$ is determined by Eq. (\ref{eq:BOE_support_mt}), we obtain Eq. (\ref{eq:structure}).
\end{proof}

This proposition implies that the structure of the matrix $T^*_a(\sigma)$ is the same as that of $T_{-a}(\sigma)$.
For deterministic strategies, in order to see this in more detail, we introduce the following sets for $a\in \{ 1, 2 \}$ and $\bm{\sigma}^{(-1)} \in \Omega$:
\begin{eqnarray}
 N_x^{(a)} \left( \bm{\sigma}^{(-1)} \right) &:=& \left\{ \left. \bm{\sigma}^{(-2)} \in \Omega \right| T_a \left( \left. C \right| \bm{\sigma}^{(-1)}, \bm{\sigma}^{(-2)} \right) = x \right\},
\end{eqnarray}
where $x\in \{ 0, 1 \}$.
That is, $N_1^{(a)} \left( \bm{\sigma}^{(-1)} \right)$ describes the set of $\bm{\sigma}^{(-2)}$ such that player $a$ using strategy $T_a$ cooperates after the history $\left[ \bm{\sigma}^{(-m)} \right]_{m=1}^{2}$.
Similarly, $N_0^{(a)} \left( \bm{\sigma}^{(-1)} \right)$ describes the set of $\bm{\sigma}^{(-2)}$ such that player $a$ using strategy $T_a$  defects after the history $\left[ \bm{\sigma}^{(-m)} \right]_{m=1}^{2}$.
We remark that $N_0^{(a)} \left( \bm{\sigma}^{(-1)} \right) \cup N_1^{(a)} \left( \bm{\sigma}^{(-1)} \right) = \Omega$ for all $a$ and $\bm{\sigma}^{(-1)}$.
Then, Proposition \ref{prop:structure} leads the following corollary:
\begin{corollary}
\label{cor:structure}
For a deterministic strategy $T_{-a}$ of player $-a$, if the optimal strategy $T^*_a$ of player $a$ against $T_{-a}$ is also deterministic, then one of the following four relations holds for each $\bm{\sigma}^{(-1)}\in \Omega$:
\begin{enumerate}[(a)]
 \item $N_x^{(a)} \left( \bm{\sigma}^{(-1)} \right) = N_x^{(-a)} \left( \bm{\sigma}^{(-1)} \right)$ for all $x$
 \item $N_x^{(a)} \left( \bm{\sigma}^{(-1)} \right) = N_{1-x}^{(-a)} \left( \bm{\sigma}^{(-1)} \right)$ for all $x$
 \item $N_0^{(a)} \left( \bm{\sigma}^{(-1)} \right) = N_0^{(-a)} \left( \bm{\sigma}^{(-1)} \right) \cup N_1^{(-a)} \left( \bm{\sigma}^{(-1)} \right) = \Omega$ and $N_1^{(a)} \left( \bm{\sigma}^{(-1)} \right) = \emptyset$
 \item $N_1^{(a)} \left( \bm{\sigma}^{(-1)} \right) = N_0^{(-a)} \left( \bm{\sigma}^{(-1)} \right) \cup N_1^{(-a)} \left( \bm{\sigma}^{(-1)} \right) = \Omega$ and $N_0^{(a)} \left( \bm{\sigma}^{(-1)} \right) = \emptyset$.
\end{enumerate}
\end{corollary}

\section{Symmetric equilibrium}
\label{sec:symmetric}
In this section, we investigate symmetric equilibrium of mutual reinforcement learning.

First, we introduce the notation $\overline{C}:=D$, $\overline{D}:=C$, and $\pi\left( \sigma_1, \sigma_2 \right) := \left( \sigma_2, \sigma_1 \right)$.
We define the word \emph{same} strategy.
\begin{definition}
\label{def:same}
A strategy $T_a$ of player $a$ is the \emph{same} strategy as that of player $-a$ iff
\begin{eqnarray}
 T_a \left( \left. \sigma \right| \bm{\sigma}^{(-1)}, \bm{\sigma}^{(-2)} \right) &=& T_{-a} \left( \left. \sigma \right| \pi\left( \bm{\sigma}^{(-1)} \right), \pi\left( \bm{\sigma}^{(-2)} \right) \right)
 \label{eq:same}
\end{eqnarray}
for all $\sigma$, $\bm{\sigma}^{(-1)}$ and, $\bm{\sigma}^{(-2)}$.
\end{definition}

Next, we introduce equilibria achieved by mutual reinforcement learning.
\begin{definition}
\label{def:RLeq}
A pair of strategy $T_1$ and $T_2$ is a \emph{mutual reinforcement learning equilibrium} iff $T_a$ is the optimal strategy against $T_{-a}$ for $a=1,2$.
\end{definition}
We emphasize that such equilibria are defined only for a time-independent part of finite-memory strategies $T_a$, although finite-memory strategies of players are generally defined as a pair of a time-independent part $T_a$ and an initial condition.
This definition is in contrast to that of Nash equilibrium or subgame perfect equilibrium.
When some appropriate initial condition is chosen, it becomes a subgame perfect equilibrium of all time-independent finite-memory strategies.
In addition, because the optimal policy is determined by comparing the action-value functions, which are functions of finite-length histories including off-equilibrium path, mutual reinforcement learning equilibrium is quite different from Nash equilibrium.

We also remark that a mutual reinforcement learning equilibrium can be achieved by Q-learning if all state-action pairs are visited an infinite number of times as mentioned above, and if an initial strategy of player 2 is appropriate.
Even if not all state-action pairs are visited an infinite number of times, we can obtain the mutual reinforcement learning equilibrium by introducing infinitesimal error probability to the opponent's strategy as in Ref. \cite{UsuUed2021}.

For deterministic mutual reinforcement learning equilibria, the following proposition is the direct consequence of Corollary \ref{cor:structure}.
\begin{proposition}
\label{prop:eq}
For mutual reinforcement learning equilibria formed by deterministic strategies, one of the following two relations holds for each $\bm{\sigma}^{(-1)} \in \Omega$:
\begin{enumerate}[(a)]
 \item $N_x^{(1)} \left( \bm{\sigma}^{(-1)} \right) = N_x^{(2)} \left( \bm{\sigma}^{(-1)} \right)$ for all $x$
 \item $N_x^{(1)} \left( \bm{\sigma}^{(-1)} \right) = N_{1-x}^{(2)} \left( \bm{\sigma}^{(-1)} \right)$ for all $x$.
\end{enumerate}
\end{proposition}

\begin{proof}
According to Corollary \ref{cor:structure}, one of the four situations (a)-(d) holds for the optimal strategy $T_1$ against $T_2$.
However, because $T_2$ is also the optimal strategy against $T_1$, the cases (c) and (d) are excluded or integrated into the case (a) or (b).
\end{proof}

Furthermore, we introduce symmetric equilibria of mutual reinforcement learning.
\begin{definition}
\label{def:symeq}
A pair of strategy $T_1$ and $T_2$ is a \emph{symmetric mutual reinforcement learning equilibrium} iff $T_a$ is the optimal strategy against $T_{-a}$ and $T_a$ is the same strategy as $T_{-a}$ for $a=1,2$.
\end{definition}

It should be noted that the deterministic optimal strategies can be written as
\begin{eqnarray}
 T^*_a\left( \left. \sigma \right| \bm{\sigma}^{(-1)}, \bm{\sigma}^{(-2)} \right) &=& \mathbb{I} \left( Q^*_a \left( \sigma, \bm{\sigma}^{(-1)}, \bm{\sigma}^{(-2)} \right) > Q^*_a \left( \overline{\sigma}, \bm{\sigma}^{(-1)}, \bm{\sigma}^{(-2)} \right) \right), \nonumber \\
 && 
\end{eqnarray}
where $\mathbb{I}(\cdots)$ is the indicator function that returns $1$ when $\cdots$ holds and $0$ otherwise.
We also introduce the following sets for $a\in \{ 1, 2 \}$ and $\bm{\sigma}^{(-1)} \in \Omega$:
\begin{eqnarray}
 \tilde{N}_x^{(a)} \left( \bm{\sigma}^{(-1)} \right) &:=& \left\{ \left. \bm{\sigma}^{(-2)} \in \Omega \right| T_a \left( \left. C \right| \bm{\sigma}^{(-1)}, \pi\left(\bm{\sigma}^{(-2)} \right) \right) = x \right\},
\end{eqnarray}
where $x\in \{ 0, 1 \}$.
We now prove the first main result of this paper.

\begin{theorem}
\label{th:symeq}
For symmetric mutual reinforcement learning equilibria formed by deterministic strategies, the following relations must hold:
\begin{enumerate}[(a)]
 \item For $\bm{\sigma}^{(-1)} \in \left\{ (C,C), (D,D) \right\}$,
 \begin{eqnarray}
 T_a \left( \left. C \right| \bm{\sigma}^{(-1)}, (C,D) \right) &=& T_a \left( \left. C \right| \bm{\sigma}^{(-1)}, (D,C) \right)
 \label{eq:symeq_property1}
 \end{eqnarray}
 for all $a$.
 \item For $\bm{\sigma}^{(-1)} \in \left\{ (C,D), (D,C) \right\}$, 
 \begin{eqnarray}
 N_x^{(a)} \left( \pi\left( \bm{\sigma}^{(-1)} \right) \right) &=& \tilde{N}_x^{(a)} \left( \bm{\sigma}^{(-1)} \right) \quad (\forall x)
 \label{eq:symeq_property2-1}
 \end{eqnarray}
 or
 \begin{eqnarray}
 N_x^{(a)} \left( \pi\left( \bm{\sigma}^{(-1)} \right) \right) &=& \tilde{N}_{1-x}^{(a)} \left( \bm{\sigma}^{(-1)} \right)  \quad (\forall x)
 \label{eq:symeq_property2-2}
 \end{eqnarray}
 holds.
\end{enumerate}
\end{theorem}

\begin{proof}
For $\bm{\sigma}^{(-1)} \in \left\{ (C,C), (D,D) \right\}$, $\pi\left( \bm{\sigma}^{(-1)}  \right) = \bm{\sigma}^{(-1)}$ holds.
Because $T_1$ and $T_2$ are the same strategies as each other, 
\begin{eqnarray}
 T_1 \left( \left. C \right| \bm{\sigma}^{(-1)}, \bm{\sigma}^{(-2)} \right) &=& T_2 \left( \left. C \right| \bm{\sigma}^{(-1)}, \bm{\sigma}^{(-2)} \right) \quad \left( \bm{\sigma}^{(-2)} \in \left\{ (C,C), (D,D) \right\} \right) \nonumber \\
 &&
\end{eqnarray}
holds.
This and Proposition \ref{prop:eq} imply that $N_x^{(1)} \left( \bm{\sigma}^{(-1)} \right)=N_x^{(2)} \left( \bm{\sigma}^{(-1)} \right)$ ($\forall x\in\{0, 1\}$) must holds.
On the other hand, due to Eq. (\ref{eq:same}),
\begin{eqnarray}
 T_1 \left( \left. C \right| \bm{\sigma}^{(-1)}, (C,D) \right) &=& T_2 \left( \left. C \right| \bm{\sigma}^{(-1)}, (D,C) \right) \\
 T_1 \left( \left. C \right| \bm{\sigma}^{(-1)}, (D,C) \right) &=& T_2 \left( \left. C \right| \bm{\sigma}^{(-1)}, (C,D) \right)
\end{eqnarray}
also hold.
This means that, if $(C,D)\in N_x^{(1)} \left( \bm{\sigma}^{(-1)} \right)$, then $(D,C)\in N_x^{(2)} \left( \pi\left( \bm{\sigma}^{(-1)} \right) \right)=N_x^{(2)} \left( \bm{\sigma}^{(-1)} \right)=N_x^{(1)} \left( \bm{\sigma}^{(-1)} \right)$, leading to Eq. (\ref{eq:symeq_property1}).

For $\bm{\sigma}^{(-1)} \in \left\{ (C,D), (D,C) \right\}$, because $T_1$ and $T_2$ are the same strategies as each other,
\begin{eqnarray}
 T_2 \left( \left. C \right| \pi\left( \bm{\sigma}^{(-1)} \right), \bm{\sigma}^{(-2)} \right) &=& T_1 \left( \left. C \right| \bm{\sigma}^{(-1)}, \pi\left( \bm{\sigma}^{(-2)} \right) \right)
\end{eqnarray}
holds for $\forall \bm{\sigma}^{(-2)} \in \Omega$.
This means that
\begin{eqnarray}
 N_x^{(2)} \left( \pi\left( \bm{\sigma}^{(-1)} \right) \right) &=& \tilde{N}_x^{(1)} \left( \bm{\sigma}^{(-1)} \right) \quad (\forall x\in\{0, 1\})
 \label{eq:same12}
\end{eqnarray}
holds.
On the other hand, Proposition \ref{prop:eq} implies that
\begin{eqnarray}
 N_x^{(1)} \left( \pi\left( \bm{\sigma}^{(-1)} \right) \right) &=& N_x^{(2)} \left( \pi\left( \bm{\sigma}^{(-1)} \right) \right) \quad (\forall x\in\{0, 1\})
 \label{eq:prop12-1}
\end{eqnarray}
or
\begin{eqnarray}
 N_x^{(1)} \left( \pi\left( \bm{\sigma}^{(-1)} \right) \right) &=& N_{1-x}^{(2)} \left( \pi\left( \bm{\sigma}^{(-1)} \right) \right) \quad (\forall x\in\{0, 1\})
 \label{eq:prop12-2}
\end{eqnarray}
must hold.
By combining Eq. (\ref{eq:same12}) and Eq. (\ref{eq:prop12-1}) or (\ref{eq:prop12-2}), we obtain Eq. (\ref{eq:symeq_property2-1}) or (\ref{eq:symeq_property2-2}).
\end{proof}

Theorem \ref{th:symeq} provides a necessary condition for a deterministic strategy to form a symmetric mutual reinforcement learning equilibrium.
In particular, Eqs. (\ref{eq:symeq_property2-1}) and (\ref{eq:symeq_property2-2}) imply that the second row and the third row of $T_a$ cannot be independent of each other.
Explicitly, $T_a$ must be one of the following $8$ forms:
\begin{eqnarray}
 \left(
\begin{array}{cccc}
a_1 & b_1 & b_1 & a_2 \\
c_1 & c_1 & c_1 & c_1 \\
d_1 & d_1 & d_1 & d_1 \\
a_3 & b_2 & b_2 & a_4
\end{array}
\right), && \left(
\begin{array}{cccc}
a_1 & b_1 & b_1 & a_2 \\
c_1 & c_1 & c_1 & 1-c_1 \\
d_1 & d_1 & d_1 & 1-d_1 \\
a_3 & b_2 & b_2 & a_4
\end{array}
\right), \nonumber \\
\left(
\begin{array}{cccc}
a_1 & b_1 & b_1 & a_2 \\
c_1 & c_1 & 1-c_1 & c_1 \\
d_1 & 1-d_1 & d_1 & d_1 \\
a_3 & b_2 & b_2 & a_4
\end{array}
\right), && \left(
\begin{array}{cccc}
a_1 & b_1 & b_1 & a_2 \\
c_1 & c_1 & 1-c_1 & 1-c_1 \\
d_1 & 1-d_1 & d_1 & 1-d_1 \\
a_3 & b_2 & b_2 & a_4
\end{array}
\right), \nonumber \\
\left(
\begin{array}{cccc}
a_1 & b_1 & b_1 & a_2 \\
c_1 & 1-c_1 & c_1 & c_1 \\
d_1 & d_1 & 1-d_1 & d_1 \\
a_3 & b_2 & b_2 & a_4
\end{array}
\right), && \left(
\begin{array}{cccc}
a_1 & b_1 & b_1 & a_2 \\
c_1 & 1-c_1 & c_1 & 1-c_1 \\
d_1 & d_1 & 1-d_1 & 1-d_1 \\
a_3 & b_2 & b_2 & a_4
\end{array}
\right), \nonumber \\
\left(
\begin{array}{cccc}
a_1 & b_1 & b_1 & a_2 \\
c_1 & 1-c_1 & 1-c_1 & c_1 \\
d_1 & 1-d_1 & 1-d_1 & d_1 \\
a_3 & b_2 & b_2 & a_4
\end{array}
\right), && \left(
\begin{array}{cccc}
a_1 & b_1 & b_1 & a_2 \\
c_1 & 1-c_1 & 1-c_1 & 1-c_1 \\
d_1 & 1-d_1 & 1-d_1 & 1-d_1 \\
a_3 & b_2 & b_2 & a_4
\end{array}
\right),
\end{eqnarray}
where $a_i, b_j, c_1, d_1 \in \{ 0, 1 \}$ $(i=1,2,3,4)$, $(j=1,2)$ independently.
For example, the Tit-for-Tat-anti-Tit-for-Tat (TFT-ATFT) strategy \cite{YBC2017}
\begin{eqnarray}
 T_1\left( C \right) &=& \left(
\begin{array}{cccc}
1 & 1 & 1 & 1 \\
0 & 0 & 0 & 1 \\
0 & 1 & 0 & 1 \\
1 & 0 & 1 & 0
\end{array}
\right),
\end{eqnarray}
does not satisfy the condition of Theorem \ref{th:symeq}, and therefore it cannot form a symmetric mutual reinforcement learning equilibrium.
However, there are still many memory-two strategies which satisfy the necessary condition, and further refinement will be needed.

\section{Examples of deterministic strategies forming symmetric equilibrium}
\label{sec:example}
In this section, we provide three examples of memory-two deterministic strategies forming symmetric mutual reinforcement learning equilibrium.
For convenience, we define the following sixteen quantities:
\begin{eqnarray}
 q_{1} &:=& R + \gamma \max_\sigma Q^*_1 \left( \sigma, (C,C), (C,C) \right) \\
 q_{2} &:=& T + \gamma \max_\sigma Q^*_1 \left( \sigma, (D,C), (C,C) \right) \\
 q_{3} &:=& S + \gamma \max_\sigma Q^*_1 \left( \sigma, (C,D), (C,C) \right) \\
 q_{4} &:=& P + \gamma \max_\sigma Q^*_1 \left( \sigma, (D,D), (C,C) \right) 
\end{eqnarray}
\begin{eqnarray}
 q_{5} &:=& R + \gamma \max_\sigma Q^*_1 \left( \sigma, (C,C), (C,D) \right) \\
 q_{6} &:=& T + \gamma \max_\sigma Q^*_1 \left( \sigma, (D,C), (C,D) \right) \\
 q_{7} &:=& S + \gamma \max_\sigma Q^*_1 \left( \sigma, (C,D), (C,D) \right) \\
 q_{8} &:=& P + \gamma \max_\sigma Q^*_1 \left( \sigma, (D,D), (C,D) \right) 
\end{eqnarray}
\begin{eqnarray}
 q_{9} &:=& R + \gamma \max_\sigma Q^*_1 \left( \sigma, (C,C), (D,C) \right) \\
 q_{10} &:=& T + \gamma \max_\sigma Q^*_1 \left( \sigma, (D,C), (D,C) \right) \\
 q_{11} &:=& S + \gamma \max_\sigma Q^*_1 \left( \sigma, (C,D), (D,C) \right) \\
 q_{12} &:=& P + \gamma \max_\sigma Q^*_1 \left( \sigma, (D,D), (D,C) \right) 
\end{eqnarray}
\begin{eqnarray}
 q_{13} &:=& R + \gamma \max_\sigma Q^*_1 \left( \sigma, (C,C), (D,D) \right) \\
 q_{14} &:=& T + \gamma \max_\sigma Q^*_1 \left( \sigma, (D,C), (D,D) \right) \\
 q_{15} &:=& S + \gamma \max_\sigma Q^*_1 \left( \sigma, (C,D), (D,D) \right) \\
 q_{16} &:=& P + \gamma \max_\sigma Q^*_1 \left( \sigma, (D,D), (D,D) \right) 
\end{eqnarray}
The Bellman optimality equation for symmetric equilibrium is
\begin{eqnarray}
 && Q^*_1\left( \sigma_1, \bm{\sigma}^{(-1)}, \bm{\sigma}^{(-2)} \right) \nonumber \\
 &=& \sum_{\sigma_2} \left\{ r_1 \left( \bm{\sigma} \right) + \max_{\hat{\sigma}} Q^*_1 \left( \hat{\sigma}, \bm{\sigma}, \bm{\sigma}^{(-1)} \right) \right\}  \nonumber \\
 && \times \mathbb{I} \left( Q^*_1 \left( \sigma_2, \pi\left( \bm{\sigma}^{(-1)} \right), \pi\left(\bm{\sigma}^{(-2)} \right) \right) > Q^*_1 \left( \overline{\sigma}_2, \pi\left( \bm{\sigma}^{(-1)} \right), \pi\left(\bm{\sigma}^{(-2)} \right) \right) \right). \nonumber \\
 &&
 \label{BOE_mt_symeq}
\end{eqnarray}
We want to find solutions of this equation.

\subsection{Delayed Grim trigger strategy}
The first candidate of the solution of Eq. (\ref{BOE_mt_symeq}) is
\begin{eqnarray}
 T_1\left( C \right) = T_2\left( C \right) &=& \left(
\begin{array}{cccc}
1 & 0 & 0 & 0 \\
1 & 0 & 0 & 0 \\
1 & 0 & 0 & 0 \\
1 & 0 & 0 & 0
\end{array}
\right).
\label{eq:rGrim}
\end{eqnarray}
We can easily check that this strategy satisfies the necessary condition for symmetric equilibrium in Theorem \ref{th:symeq}.
Because this strategy is a variant of the Grim trigger strategy \cite{Fri1971}
\begin{eqnarray}
 T_1\left( C \right) &=& \left(
\begin{array}{cccc}
1 & 1 & 1 & 1 \\
0 & 0 & 0 & 0 \\
0 & 0 & 0 & 0 \\
0 & 0 & 0 & 0
\end{array}
\right)
\label{eq:Grim}
\end{eqnarray}
but uses only information at the second last action profile, the strategy (\ref{eq:rGrim}) can be called as \emph{delayed Grim} strategy.
\begin{theorem}
\label{th:rGrim}
A pair of the strategy (\ref{eq:rGrim}) forms a symmetric mutual reinforcement learning equilibrium if $\gamma>\sqrt{\frac{T-R}{T-P}}$.
\end{theorem}

\begin{proof}
The Bellman optimality equation against the strategy (\ref{eq:rGrim}) is
\begin{eqnarray}
 Q^*_1 \left( C, (C,C), (C,C) \right) &=& q_{1} \\
 Q^*_1 \left( D, (C,C), (C,C) \right) &=& q_{2} \\
 Q^*_1 \left( C, (C,C), \bm{\sigma}^{(-2)} \right) &=& q_{3} \quad \left( \bm{\sigma}^{(-2)} \neq (C,C) \right) \\
 Q^*_1 \left( D, (C,C), \bm{\sigma}^{(-2)}  \right) &=& q_{4} \quad \left( \bm{\sigma}^{(-2)} \neq (C,C) \right)
\end{eqnarray}
\begin{eqnarray}
 Q^*_1 \left( C, (C,D), (C,C) \right) &=& q_{5} \\
 Q^*_1 \left( D, (C,D), (C,C) \right) &=& q_{6} \\
 Q^*_1 \left( C, (C,D), \bm{\sigma}^{(-2)} \right) &=& q_{7} \quad \left( \bm{\sigma}^{(-2)} \neq (C,C) \right) \\
 Q^*_1 \left( D, (C,D), \bm{\sigma}^{(-2)}  \right) &=& q_{8} \quad \left( \bm{\sigma}^{(-2)} \neq (C,C) \right)
\end{eqnarray}
\begin{eqnarray}
 Q^*_1 \left( C, (D,C), (C,C) \right) &=& q_{9} \\
 Q^*_1 \left( D, (D,C), (C,C) \right) &=& q_{10} \\
 Q^*_1 \left( C, (D,C), \bm{\sigma}^{(-2)} \right) &=& q_{11} \quad \left( \bm{\sigma}^{(-2)} \neq (C,C) \right) \\
 Q^*_1 \left( D, (D,C), \bm{\sigma}^{(-2)}  \right) &=& q_{12} \quad \left( \bm{\sigma}^{(-2)} \neq (C,C) \right)
\end{eqnarray}
\begin{eqnarray}
 Q^*_1 \left( C, (D,D), (C,C) \right) &=& q_{13} \\
 Q^*_1 \left( D, (D,D), (C,C) \right) &=& q_{14} \\
 Q^*_1 \left( C, (D,D), \bm{\sigma}^{(-2)} \right) &=& q_{15} \quad \left( \bm{\sigma}^{(-2)} \neq (C,C) \right) \\
 Q^*_1 \left( D, (D,D), \bm{\sigma}^{(-2)}  \right) &=& q_{16} \quad \left( \bm{\sigma}^{(-2)} \neq (C,C) \right)
\end{eqnarray}
with the self-consistency condition
\begin{eqnarray}
 q_{1} &>& q_{2} \nonumber \\
 q_{3} &<& q_{4} \nonumber \\
 q_{5} &>& q_{6} \nonumber \\
 q_{7} &<& q_{8} \nonumber \\
 q_{9} &>& q_{10} \nonumber \\
 q_{11} &<& q_{12} \nonumber \\
 q_{13} &>& q_{14} \nonumber \\
 q_{15} &<& q_{16}.
 \label{eq:self-consistent_rGrim}
\end{eqnarray}
The solution is
\begin{eqnarray}
 q_{1} &=& \frac{1}{1-\gamma} R \\
 q_{2} &=& T + \frac{\gamma}{1-\gamma^2} R + \frac{\gamma^2}{1-\gamma^2} P \\
 q_{3} &=& S + \frac{\gamma}{1-\gamma^2} R + \frac{\gamma^2}{1-\gamma^2} P \\
 q_{4} &=& \frac{1}{1-\gamma^2} P + \frac{\gamma}{1-\gamma^2} R \\
 q_{5} = q_{9} = q_{13} &=& \frac{1}{1-\gamma^2} R + \frac{\gamma}{1-\gamma^2} P \\
 q_{6} = q_{10} = q_{14} &=& T + \frac{\gamma}{1-\gamma} P \\
 q_{7} = q_{11} = q_{15} &=& S + \frac{\gamma}{1-\gamma} P \\
 q_{8} = q_{12} = q_{16} &=& \frac{1}{1-\gamma} P.
\end{eqnarray}
For these solution, the inequalities (\ref{eq:self-consistent_rGrim}) are satisfied if
\begin{eqnarray}
 \gamma &>& \sqrt{\frac{T-R}{T-P}}.
 \label{eq:condition_rGrim}
\end{eqnarray}
\end{proof}

We remark that the condition (\ref{eq:condition_rGrim}) is more strict than the condition that Grim forms a symmetric equilibrium \cite{UsuUed2021}: $\gamma>\frac{T-R}{T-P}$.

\subsection{Delayed Win-Stay Lose-Shift strategy}
The second candidate of the solution of Eq. (\ref{BOE_mt_symeq}) is
\begin{eqnarray}
 T_1\left( C \right) = T_2\left( C \right) &=& \left(
\begin{array}{cccc}
1 & 0 & 0 & 1 \\
1 & 0 & 0 & 1 \\
1 & 0 & 0 & 1 \\
1 & 0 & 0 & 1
\end{array}
\right).
\label{eq:rWSLS}
\end{eqnarray}
We can easily check that this strategy satisfies the necessary condition for symmetric equilibrium in Theorem \ref{th:symeq}.
Because this strategy is a variant of the Win-Stay Lose-Shift (WSLS) strategy \cite{NowSig1993}
\begin{eqnarray}
 T_1\left( C \right) &=& \left(
\begin{array}{cccc}
1 & 1 & 1 & 1 \\
0 & 0 & 0 & 0 \\
0 & 0 & 0 & 0 \\
1 & 1 & 1 & 1
\end{array}
\right)
\label{eq:WSLS}
\end{eqnarray}
but uses only information at the second last action profile, the strategy (\ref{eq:rWSLS}) can be called as \emph{delayed WSLS} strategy.
\begin{theorem}
\label{th:rWSLS}
When $2R>T+P$ holds, a pair of the strategy (\ref{eq:rWSLS}) forms a symmetric mutual reinforcement learning equilibrium if $\gamma>\sqrt{\frac{T-R}{R-P}}$.
\end{theorem}

\begin{proof}
The Bellman optimality equation against the strategy (\ref{eq:rWSLS}) is
\begin{eqnarray}
 Q^*_1 \left( C, (C,C), \bm{\sigma}^{(-2)} \right) &=& q_{1} \quad \left( \bm{\sigma}^{(-2)} \in \{ (C,C), (D,D) \} \right) \\
 Q^*_1 \left( D, (C,C), \bm{\sigma}^{(-2)} \right) &=& q_{2} \quad \left( \bm{\sigma}^{(-2)} \in \{ (C,C), (D,D) \} \right) \\
 Q^*_1 \left( C, (C,C), \bm{\sigma}^{(-2)} \right) &=& q_{3} \quad \left( \bm{\sigma}^{(-2)} \in \{ (C,D), (D,C) \} \right) \\
 Q^*_1 \left( D, (C,C), \bm{\sigma}^{(-2)}  \right) &=& q_{4} \quad \left( \bm{\sigma}^{(-2)} \in \{ (C,D), (D,C) \} \right)
\end{eqnarray}
\begin{eqnarray}
 Q^*_1 \left( C, (C,D), \bm{\sigma}^{(-2)} \right) &=& q_{5} \quad \left( \bm{\sigma}^{(-2)} \in \{ (C,C), (D,D) \} \right) \\
 Q^*_1 \left( D, (C,D), \bm{\sigma}^{(-2)} \right) &=& q_{6} \quad \left( \bm{\sigma}^{(-2)} \in \{ (C,C), (D,D) \} \right) \\
 Q^*_1 \left( C, (C,D), \bm{\sigma}^{(-2)} \right) &=& q_{7} \quad \left( \bm{\sigma}^{(-2)} \in \{ (C,D), (D,C) \} \right) \\
 Q^*_1 \left( D, (C,D), \bm{\sigma}^{(-2)}  \right) &=& q_{8} \quad \left( \bm{\sigma}^{(-2)} \in \{ (C,D), (D,C) \} \right)
\end{eqnarray}
\begin{eqnarray}
 Q^*_1 \left( C, (D,C), \bm{\sigma}^{(-2)} \right) &=& q_{9} \quad \left( \bm{\sigma}^{(-2)} \in \{ (C,C), (D,D) \} \right) \\
 Q^*_1 \left( D, (D,C), \bm{\sigma}^{(-2)} \right) &=& q_{10} \quad \left( \bm{\sigma}^{(-2)} \in \{ (C,C), (D,D) \} \right) \\
 Q^*_1 \left( C, (D,C), \bm{\sigma}^{(-2)} \right) &=& q_{11} \quad \left( \bm{\sigma}^{(-2)} \in \{ (C,D), (D,C) \} \right) \\
 Q^*_1 \left( D, (D,C), \bm{\sigma}^{(-2)}  \right) &=& q_{12} \quad \left( \bm{\sigma}^{(-2)} \in \{ (C,D), (D,C) \} \right)
\end{eqnarray}
\begin{eqnarray}
 Q^*_1 \left( C, (D,D), \bm{\sigma}^{(-2)} \right) &=& q_{13} \quad \left( \bm{\sigma}^{(-2)} \in \{ (C,C), (D,D) \} \right) \\
 Q^*_1 \left( D, (D,D), \bm{\sigma}^{(-2)} \right) &=& q_{14} \quad \left( \bm{\sigma}^{(-2)} \in \{ (C,C), (D,D) \} \right) \\
 Q^*_1 \left( C, (D,D), \bm{\sigma}^{(-2)} \right) &=& q_{15} \quad \left( \bm{\sigma}^{(-2)} \in \{ (C,D), (D,C) \} \right) \\
 Q^*_1 \left( D, (D,D), \bm{\sigma}^{(-2)}  \right) &=& q_{16} \quad \left( \bm{\sigma}^{(-2)} \in \{ (C,D), (D,C) \} \right)
\end{eqnarray}
with the self-consistency condition
\begin{eqnarray}
 q_{1} &>& q_{2} \nonumber \\
 q_{3} &<& q_{4} \nonumber \\
 q_{5} &>& q_{6} \nonumber \\
 q_{7} &<& q_{8} \nonumber \\
 q_{9} &>& q_{10} \nonumber \\
 q_{11} &<& q_{12} \nonumber \\
 q_{13} &>& q_{14} \nonumber \\
 q_{15} &<& q_{16}.
 \label{eq:self-consistent_rWSLS}
\end{eqnarray}
The solution is
\begin{eqnarray}
 q_{1} = q_{13} &=& \frac{1}{1-\gamma} R \\
 q_{2} = q_{14} &=& T + \gamma R + \gamma^2 P + \frac{\gamma^3}{1-\gamma} R \\
 q_{3} = q_{15} &=& S + \gamma R + \gamma^2 P + \frac{\gamma^3}{1-\gamma} R \\
 q_{4} = q_{16} &=& P + \frac{\gamma}{1-\gamma} R \\
 q_{5} = q_{9} &=& R + \gamma P + \frac{\gamma^2}{1-\gamma} R \\
 q_{6} = q_{10} &=& T + \gamma P + \gamma^2 P + \frac{\gamma^3}{1-\gamma} R \\
 q_{7} = q_{11} &=& S + \gamma P + \gamma^2 P + \frac{\gamma^3}{1-\gamma} R \\
 q_{8} = q_{12} &=& P + \gamma P + \frac{\gamma^2}{1-\gamma} R.
\end{eqnarray}
For these solution, the inequalities (\ref{eq:self-consistent_rWSLS}) are satisfied if
\begin{eqnarray}
 \gamma &>& \sqrt{\frac{T-R}{R-P}}.
 \label{eq:condition_rWSLS}
\end{eqnarray}
It should be noted that such $\gamma<1$ exists only if $2R>T+P$.
\end{proof}

We remark that the condition (\ref{eq:condition_rWSLS}) is more strict than the condition that WSLS forms a symmetric equilibrium \cite{UsuUed2021}: $\gamma>\frac{T-R}{R-P}$.

\subsection{All-or-None strategy}
The third example of the solution of Eq. (\ref{BOE_mt_symeq}) is the All-or-None strategy $AON_2$ \cite{HMCN2017}:
\begin{eqnarray}
 T_1\left( C \right) = T_2\left( C \right) &=& \left(
\begin{array}{cccc}
1 & 0 & 0 & 1 \\
0 & 0 & 0 & 0 \\
0 & 0 & 0 & 0 \\
1 & 0 & 0 & 1
\end{array}
\right).
\label{eq:AON2}
\end{eqnarray}
We can easily check that this strategy satisfies the necessary condition for symmetric equilibrium in Theorem \ref{th:symeq}.
It has been known that $AON_2$ forms subgame perfect equilibrium \cite{HMCN2017}.
A similar strategy as $AON_2$ was also observed in numerical simulation of evolution of cooperation \cite{HauSch1997}.

\begin{theorem}
\label{th:AON2}
When $3R-2P-T>0$ and $2R-3P+S>0$ hold, a pair of the strategy (\ref{eq:AON2}) forms a symmetric mutual reinforcement learning equilibrium if
\begin{eqnarray}
 \gamma &>& \max \left\{ \frac{1}{2} \left( \sqrt{\frac{4T-3R-P}{R-P}} - 1 \right), \frac{1}{2} \left( \sqrt{\frac{R+3P-4S}{R-P}} - 1 \right) \right\}.
 \label{eq:condition_AON2}
\end{eqnarray}
\end{theorem}

\begin{proof}
The Bellman optimality equation against the strategy (\ref{eq:AON2}) is
\begin{eqnarray}
 Q^*_1 \left( C, (C,C), \bm{\sigma}^{(-2)} \right) &=& q_{1} \quad \left( \bm{\sigma}^{(-2)} \in \{ (C,C), (D,D) \} \right) \\
 Q^*_1 \left( D, (C,C), \bm{\sigma}^{(-2)} \right) &=& q_{2} \quad \left( \bm{\sigma}^{(-2)} \in \{ (C,C), (D,D) \} \right) \\
 Q^*_1 \left( C, (C,C), \bm{\sigma}^{(-2)} \right) &=& q_{3} \quad \left( \bm{\sigma}^{(-2)} \in \{ (C,D), (D,C) \} \right) \\
 Q^*_1 \left( D, (C,C), \bm{\sigma}^{(-2)}  \right) &=& q_{4} \quad \left( \bm{\sigma}^{(-2)} \in \{ (C,D), (D,C) \} \right)
\end{eqnarray}
\begin{eqnarray}
 Q^*_1 \left( C, (C,D), \bm{\sigma}^{(-2)} \right) &=& q_{7} \quad \left( \bm{\sigma}^{(-2)} \in \Omega \right) \\
 Q^*_1 \left( D, (C,D), \bm{\sigma}^{(-2)}  \right) &=& q_{8} \quad \left( \bm{\sigma}^{(-2)} \in \Omega \right)
\end{eqnarray}
\begin{eqnarray}
 Q^*_1 \left( C, (D,C), \bm{\sigma}^{(-2)} \right) &=& q_{11} \quad \left( \bm{\sigma}^{(-2)} \in \Omega \right) \\
 Q^*_1 \left( D, (D,C), \bm{\sigma}^{(-2)}  \right) &=& q_{12} \quad \left( \bm{\sigma}^{(-2)} \in \Omega \right)
\end{eqnarray}
\begin{eqnarray}
 Q^*_1 \left( C, (D,D), \bm{\sigma}^{(-2)} \right) &=& q_{13} \quad \left( \bm{\sigma}^{(-2)} \in \{ (C,C), (D,D) \} \right) \\
 Q^*_1 \left( D, (D,D), \bm{\sigma}^{(-2)} \right) &=& q_{14} \quad \left( \bm{\sigma}^{(-2)} \in \{ (C,C), (D,D) \} \right) \\
 Q^*_1 \left( C, (D,D), \bm{\sigma}^{(-2)} \right) &=& q_{15} \quad \left( \bm{\sigma}^{(-2)} \in \{ (C,D), (D,C) \} \right) \\
 Q^*_1 \left( D, (D,D), \bm{\sigma}^{(-2)}  \right) &=& q_{16} \quad \left( \bm{\sigma}^{(-2)} \in \{ (C,D), (D,C) \} \right)
\end{eqnarray}
with the self-consistency condition
\begin{eqnarray}
 q_{1} &>& q_{2} \nonumber \\
 q_{3} &<& q_{4} \nonumber \\
 q_{7} &<& q_{8} \nonumber \\
 q_{11} &<& q_{12} \nonumber \\
 q_{13} &>& q_{14} \nonumber \\
 q_{15} &<& q_{16}.
 \label{eq:self-consistent_AON2}
\end{eqnarray}
The solution is
\begin{eqnarray}
 q_{1} = q_{13} &=& \frac{1}{1-\gamma} R \\
 q_{2} = q_{14} &=& T + \gamma P + \gamma^2 P + \frac{\gamma^3}{1-\gamma} R \\
 q_{3} = q_{7} = q_{11} = q_{15} &=& S + \gamma P + \gamma^2 P + \frac{\gamma^3}{1-\gamma} R \\
 q_{4} = q_{16} &=& P + \frac{\gamma}{1-\gamma} R \\
 q_{8} = q_{12} &=& P + \gamma P + \frac{\gamma^2}{1-\gamma} R.
\end{eqnarray}
For these solution, the inequalities (\ref{eq:self-consistent_AON2}) are satisfied if
\begin{eqnarray}
 (R-P) \gamma^2 + (R-P) \gamma - (T-R) &>& 0
\end{eqnarray}
and
\begin{eqnarray}
 (R-P) \gamma^2 + (R-P) \gamma - (P-S) &>& 0,
\end{eqnarray}
which are equivalent to Eq. (\ref{eq:condition_AON2}) for $\gamma \geq 0$.
It should be noted that such $\gamma<1$ exists only if $3R-2P-T>0$ and $2R-3P+S>0$.
\end{proof}

\section{Optimality in longer memory}
\label{sec:optimality}
In previous sections, we investigated symmetric equilibrium of mutual reinforcement learning when both players use memory-two strategies, and obtained three examples of deterministic strategies forming symmetric equilibrium.
A natural question is ``Do these strategies forming symmetric equilibrium in memory-two reinforcement learning also form symmetric equilibrium of mutual reinforcement learning of longer memory strategies?''.
In this section, we show that the answer to this question is ``yes''.

We first prove the following theorem.
\begin{theorem}
\label{th:longer}
Let $T_{-a}$ be a memory-$n^\prime$ strategy of player $-a$.
Let $T^*_a$ be the optimal strategy of player $a$ against $T_{-a}$ when player $a$ use reinforcement learning of memory-$n^\prime$ strategies.
When player $a$ use reinforcement learning of memory-$n$ strategies with $n>n^\prime$ to obtain the optimal strategy $\check{T}^*_a$ against $T_{-a}$, then $\check{T}^*_a=T^*_a$.
\end{theorem}

\begin{proof}
When player $-a$ use memory-$n^\prime$ strategy and player $a$ use memory-$n$ reinforcement learning with $n>n^\prime$, the Bellman optimality equation (\ref{eq:bellman}) becomes
\begin{eqnarray}
 && Q^*_a\left( \sigma_a, \left[ \bm{\sigma}^{(-m)} \right]_{m=1}^{n} \right) \nonumber \\
 &=& \sum_{\sigma_{-a}} r_a \left( \bm{\sigma} \right) T_{-a} \left( \left. \sigma_{-a} \right| \left[ \bm{\sigma}^{(-m)} \right]_{m=1}^{n^\prime} \right) \nonumber \\
 && + \gamma \sum_{\sigma_{-a}} T_{-a} \left( \left. \sigma_{-a} \right| \left[ \bm{\sigma}^{(-m)} \right]_{m=1}^{n^\prime} \right) \max_{\hat{\sigma}} Q^*_a \left( \hat{\sigma}, \bm{\sigma}, \left[ \bm{\sigma}^{(-m)} \right]_{m=1}^{n-1} \right). \nonumber \\
 &&
\end{eqnarray}
Then, we find that the right-hand side does not depend on $\bm{\sigma}^{(-n)}$, and therefore
\begin{eqnarray}
 Q^*_a\left( \sigma_a, \left[ \bm{\sigma}^{(-m)} \right]_{m=1}^{n} \right) &=& Q^*_a\left( \sigma_a, \left[ \bm{\sigma}^{(-m)} \right]_{m=1}^{n-1} \right)
\end{eqnarray}
Then, the Bellman optimality equation becomes
\begin{eqnarray}
 && Q^*_a\left( \sigma_a, \left[ \bm{\sigma}^{(-m)} \right]_{m=1}^{n-1} \right) \nonumber \\
 &=& \sum_{\sigma_{-a}} r_a \left( \bm{\sigma} \right) T_{-a} \left( \left. \sigma_{-a} \right| \left[ \bm{\sigma}^{(-m)} \right]_{m=1}^{n^\prime} \right) \nonumber \\
 && + \gamma \sum_{\sigma_{-a}} T_{-a} \left( \left. \sigma_{-a} \right| \left[ \bm{\sigma}^{(-m)} \right]_{m=1}^{n^\prime} \right) \max_{\hat{\sigma}} Q^*_a \left( \hat{\sigma}, \bm{\sigma}, \left[ \bm{\sigma}^{(-m)} \right]_{m=1}^{n-2} \right). \nonumber \\
 &&
\end{eqnarray}
By repeating the same argument until the length of memory decreases to $n^\prime$, we find that
\begin{eqnarray}
 Q^*_a\left( \sigma_a, \left[ \bm{\sigma}^{(-m)} \right]_{m=1}^{n} \right) &=& Q^*_a\left( \sigma_a, \left[ \bm{\sigma}^{(-m)} \right]_{m=1}^{n^\prime} \right),
\end{eqnarray}
which implies that $\check{T}^*_a=T^*_a$.
\end{proof}
That is, when the opponent $-a$ uses a memory-two strategy $T_{-a}$, and player $a$ learns the optimal memory-$n$ strategy with $n\geq 2$ against $T_{-a}$, then, such optimal strategy is memory-two.
Similarly, when the opponent $-a$ uses a memory-one strategy, and player $a$ learns the optimal memory-$n$ strategy with $n\geq 1$, then, such optimal strategy is memory-one.

This theorem results in the following corollary.
\begin{corollary}
\label{cor:eq_longer}
A mutual reinforcement learning equilibrium obtained by memory-$n^\prime$ reinforcement learning is also a mutual reinforcement learning equilibrium obtained by memory-$n$ reinforcement learning with $n>n^\prime$.
\end{corollary}

Therefore, the strategies (\ref{eq:rGrim}) and (\ref{eq:rWSLS}) in the previous section also form mutual reinforcement learning equilibria even if players use memory-$n$ reinforcement learning with $n>2$.
Similarly, the (memory-one) Grim strategy and the (memory-one) WSLS strategy still form mutual reinforcement learning equilibria even if players use memory-two reinforcement learning, since it has been known that Grim and WSLS form memory-one mutual reinforcement learning equilibria, respectively \cite{UsuUed2021}.
We remark that this property is similar to that of Nash equilibrium in finite automaton selection game, which claims that two automata must have an equal number of states in equilibria \cite{AbrRub1988}.

\section{Conclusion}
\label{sec:conclusion}
In this paper, we investigated symmetric equilibrium of mutual reinforcement learning when both players use memory-two deterministic strategies in the repeated prisoners' dilemma game.
First, we find that the structure of the optimal strategies is constrained by the Bellman optimality equation (Proposition \ref{prop:structure}).
Then, we find a necessary condition for deterministic symmetric equilibrium of mutual reinforcement learning (Theorem \ref{th:symeq}).
Furthermore, we provided three examples of memory-two deterministic strategies which form symmetric mutual reinforcement learning equilibrium, some of which can be regarded as variants of the Grim strategy and the WSLS strategy (Theorem \ref{th:rGrim}, Theorem \ref{th:rWSLS} and Theorem \ref{th:AON2}).
Finally, we proved that mutual reinforcement learning equilibria achieved by memory-two strategies are also mutual reinforcement learning equilibria when both players use reinforcement learning of memory-$n$ strategies with $n>2$ (Theorem \ref{th:longer}).

We want to investigate whether other symmetric mutual reinforcement learning equilibria of deterministic memory-two strategies exist or not in future.
For such purpose, novel methods are needed, because the number of strategies is quite large.
Furthermore, extension of our analysis to (i) asymmetric equilibrium and (ii) mixed strategies is also a subject of future work.
Ultimately, if we would develop some method to find all equilibria in all length of memory $n$, analysis of mutual reinforcement learning equilibria is completed.

\section*{Acknowledgement}
We thank Genki Ichinose and Mashiho Mukaida for useful discussions.
This study was supported by JSPS KAKENHI Grant Number JP20K19884 and Inamori Research Grants.

\appendix
\section{Derivation of Eq. (\ref{eq:BE})}
\label{app:BE}
First we introduce the notation
\begin{eqnarray}
 T\left( \left. \bm{\sigma} \right| \left[ \bm{\sigma}^{(-m)} \right]_{m=1}^{n} \right) &:=& \prod_{a=1}^2 T_a\left( \left. \sigma_a \right| \left[ \bm{\sigma}^{(-m)} \right]_{m=1}^{n} \right).
\end{eqnarray}
We remark that the joint probability distribution of the action profiles $\bm{\sigma}^{(\mu)}$, $\cdots$, $\bm{\sigma}^{(0)}$ given $\left[ \bm{\sigma}^{(-m)} \right]_{m=1}^{n}$ is described as
\begin{eqnarray}
 P\left( \bm{\sigma}^{(\mu)}, \cdots, \bm{\sigma}^{(0)} \left| \left[ \bm{\sigma}^{(-m)} \right]_{m=1}^{n} \right. \right) &=& \prod_{s=0}^\mu T\left( \left. \bm{\sigma}^{(s)} \right| \left[ \bm{\sigma}^{(s-m)} \right]_{m=1}^{n} \right).
\end{eqnarray}
The action-value function (\ref{eq:def_Q}) is rewritten as
\begin{eqnarray}
&& Q_a\left( \sigma^{(0)}_a, \left[ \bm{\sigma}^{(-m)} \right]_{m=1}^{n} \right) \nonumber \\
&=& \sum_{\left[ \bm{\sigma}^{(s)} \right]_{s=1}^{\infty}} \sum_{\sigma^{(0)}_{-a}} \sum_{k=0}^\infty \gamma^k r_a\left( \bm{\sigma}^{(k)} \right) \left\{ \prod_{s=1}^\infty T\left( \left. \bm{\sigma}^{(s)} \right| \left[ \bm{\sigma}^{(s-m)} \right]_{m=1}^{n} \right) \right\} \nonumber \\
 && \times T_{-a}\left( \left. \sigma^{(0)}_{-a} \right| \left[ \bm{\sigma}^{(-m)} \right]_{m=1}^{n} \right) \nonumber \\
 &=& \sum_{\left[ \bm{\sigma}^{(s)} \right]_{s=1}^{\infty}} \sum_{\sigma^{(0)}_{-a}} \left[ r_a\left( \bm{\sigma}^{(0)} \right) + \gamma \sum_{k=0}^\infty \gamma^k r_a\left( \bm{\sigma}^{(k+1)} \right) \right] \nonumber \\
 && \times \left\{ \prod_{s=1}^\infty T\left( \left. \bm{\sigma}^{(s)} \right| \left[ \bm{\sigma}^{(s-m)} \right]_{m=1}^{n} \right) \right\} T_{-a}\left( \left. \sigma^{(0)}_{-a} \right| \left[ \bm{\sigma}^{(-m)} \right]_{m=1}^{n} \right) \nonumber \\
 &=& \sum_{\sigma^{(0)}_{-a}} r_a\left( \bm{\sigma}^{(0)} \right) T_{-a}\left( \left. \sigma^{(0)}_{-a} \right| \left[ \bm{\sigma}^{(-m)} \right]_{m=1}^{n} \right) \nonumber \\
 && + \gamma \sum_{\left[ \bm{\sigma}^{(s)} \right]_{s=1}^{\infty}} \sum_{\sigma^{(0)}_{-a}} \sum_{k=0}^\infty \gamma^k r_a\left( \bm{\sigma}^{(k+1)} \right) \left\{ \prod_{s=2}^\infty T\left( \left. \bm{\sigma}^{(s)} \right| \left[ \bm{\sigma}^{(s-m)} \right]_{m=1}^{n} \right) \right\} \nonumber \\
 && \times T_{-a}\left( \left. \sigma^{(1)}_{-a} \right| \left[ \bm{\sigma}^{(1-m)} \right]_{m=1}^{n} \right) T_{a}\left( \left. \sigma^{(1)}_{a} \right| \left[ \bm{\sigma}^{(1-m)} \right]_{m=1}^{n} \right) T_{-a}\left( \left. \sigma^{(0)}_{-a} \right| \left[ \bm{\sigma}^{(-m)} \right]_{m=1}^{n} \right) \nonumber \\
 &=& \sum_{\sigma^{(0)}_{-a}} r_a\left( \bm{\sigma}^{(0)} \right) T_{-a}\left( \left. \sigma^{(0)}_{-a} \right| \left[ \bm{\sigma}^{(-m)} \right]_{m=1}^{n} \right) \nonumber \\
 && + \gamma \sum_{\sigma^{(1)}_{a}} \sum_{\sigma^{(0)}_{-a}} Q_a\left( \sigma^{(1)}_a, \left[ \bm{\sigma}^{(1-m)} \right]_{m=1}^{n} \right) T_{a}\left( \left. \sigma^{(1)}_{a} \right| \left[ \bm{\sigma}^{(1-m)} \right]_{m=1}^{n} \right) T_{-a}\left( \left. \sigma^{(0)}_{-a} \right| \left[ \bm{\sigma}^{(-m)} \right]_{m=1}^{n} \right), \nonumber \\
 &&
\end{eqnarray}
which is Eq. (\ref{eq:BE}).

\section{Derivation of Eqs. (\ref{eq:bellman}) and (\ref{eq:BOE_support})}
\label{app:BOE}
We define $Q^*_a$ as the optimal value of $Q_a$, which is obtained by choosing optimal policy $T^*_a$.
Then, $Q^*_a$ obeys
\begin{eqnarray}
 && Q^*_a\left( \sigma_a, \left[ \bm{\sigma}^{(-m)} \right]_{m=1}^{n} \right) \nonumber \\
 &=& \sum_{\sigma_{-a}} r_a \left( \bm{\sigma} \right) T_{-a} \left( \left. \sigma_{-a} \right| \left[ \bm{\sigma}^{(-m)} \right]_{m=1}^{n} \right) \nonumber \\
 && + \gamma \sum_{\sigma_{a}^{\prime}} \sum_{\sigma_{-a}} T^*_{a} \left( \left. \sigma_{a}^{\prime} \right| \bm{\sigma}, \left[ \bm{\sigma}^{(-m)} \right]_{m=1}^{n-1} \right) T_{-a} \left( \left. \sigma_{-a} \right| \left[ \bm{\sigma}^{(-m)} \right]_{m=1}^{n} \right) Q^*_a \left( \sigma_{a}^{\prime}, \bm{\sigma}, \left[ \bm{\sigma}^{(-m)} \right]_{m=1}^{n-1} \right) \nonumber \\
 &\leq& \sum_{\sigma_{-a}} r_a \left( \bm{\sigma} \right) T_{-a} \left( \left. \sigma_{-a} \right| \left[ \bm{\sigma}^{(-m)} \right]_{m=1}^{n} \right) \nonumber \\
 && + \gamma \sum_{\sigma_{a}^{\prime}} \sum_{\sigma_{-a}} T^*_{a} \left( \left. \sigma_{a}^{\prime} \right| \bm{\sigma}, \left[ \bm{\sigma}^{(-m)} \right]_{m=1}^{n-1} \right) T_{-a} \left( \left. \sigma_{-a} \right| \left[ \bm{\sigma}^{(-m)} \right]_{m=1}^{n} \right) \max_{\hat{\sigma}} Q^*_a \left( \hat{\sigma}, \bm{\sigma}, \left[ \bm{\sigma}^{(-m)} \right]_{m=1}^{n-1} \right) \nonumber \\
 &=& \sum_{\sigma_{-a}} r_a \left( \bm{\sigma} \right) T_{-a} \left( \left. \sigma_{-a} \right| \left[ \bm{\sigma}^{(-m)} \right]_{m=1}^{n} \right) \nonumber \\
 && + \gamma \sum_{\sigma_{-a}} T_{-a} \left( \left. \sigma_{-a} \right| \left[ \bm{\sigma}^{(-m)} \right]_{m=1}^{n} \right) \max_{\hat{\sigma}} Q^*_a \left( \hat{\sigma}, \bm{\sigma}, \left[ \bm{\sigma}^{(-m)} \right]_{m=1}^{n-1} \right).
\end{eqnarray}
The equality in the third line holds when
\begin{eqnarray}
 \mathrm{supp}T^*_a\left( \left. \cdot \right| \left[ \bm{\sigma}^{(-m)} \right]_{m=1}^{n} \right) &=& \arg \max_{\sigma} Q^*_a \left( \sigma, \left[ \bm{\sigma}^{(-m)} \right]_{m=1}^{n} \right),
\end{eqnarray}
which is Eq. (\ref{eq:BOE_support}).

\section*{References}

\bibliography{mtRL}

\end{document}